\documentclass[conference]{IEEEtran}
\IEEEoverridecommandlockouts

\usepackage{cite}
\usepackage{amsmath,amssymb,amsfonts}
\usepackage{textcomp}
\usepackage{xcolor}
\usepackage{xurl}

\usepackage{amsthm}

\usepackage{algorithm}
\usepackage{algpseudocode}

\usepackage{graphicx}
\usepackage{booktabs}       
\usepackage{multirow}
\usepackage{makecell}       
\usepackage{threeparttable} 
\usepackage{graphbox}       

\usepackage{subcaption}
\DeclareCaptionLabelSeparator{periodspace}{.\quad}
\usepackage{pifont}         
\usepackage{wasysym}
\usepackage{circledsteps}
\usepackage{enumitem}       
\usepackage{fancybox}       
\usepackage[style=english]{csquotes}
\usepackage[hidelinks]{hyperref} 

\newtheorem{theorem}{Theorem}
\newtheorem{lemma}{Lemma}

\newcommand{\protocolname}{\textit{OrbitBFT}}

\def\BibTeX{{\rm B\kern-.05em{\sc i\kern-.025em b}\kern-.08em
    T\kern-.1667em\lower.7ex\hbox{E}\kern-.125emX}}

\begin{document}

\title{OrbitBFT: Enabling Scalable and Robust BFT Consensus in LEO Constellations

\thanks{
Corresponding author: Minghui~Xu (\href{mailto:mhxu@sdu.edu.cn}{mhxu@sdu.edu.cn})

This study was supported by the National Natural Science Foundation of China (No. 62302266, U25A20425, 62232010, U23A20302, U24A20244), the Research Project of Quancheng Laboratory, China under Grant No. QCL20250106, and 2025 Shandong Cyberspace Administration Reform Pilot Projects.
}
}

\author{
\IEEEauthorblockN{
Tianyi~Sun$^{\dag*}$, Shuo~Liu$^{\dag*}$, Minghui~Xu$^{\ddag\dag}$, Xiuzhen~Cheng$^{\dag}$
}
\IEEEauthorblockA{
$^\dag$ School of Computer Science and Technology, Shandong University\\
$^\ddag$ Quancheng Laboratory\\
*Both authors contributed equally to this work
}
}


\maketitle

\begin{abstract}
Low Earth Orbit (LEO) satellite constellations are evolving from communication relays into autonomous platforms operating in increasingly congested and contested environments. Since uplinks to ground stations can be severed or jammed, ensuring reliable coordination among satellites requires autonomous Byzantine Fault-Tolerant (BFT) consensus. However, applying conventional BFT protocols to LEO constellations is challenging due to their dynamic topology, sparse connectivity, and limited communication bandwidth. In this paper, we present OrbitBFT, a novel two-stage hierarchical BFT consensus protocol tailored to the unique characteristics of LEO constellations. First, OrbitBFT exploits the topological stability within orbital planes to partition the constellation and perform localized intra-plane consensus, which reduces communication overhead. Second, we design a Byzantine-resilient bypass mechanism and a hop-by-hop transmission protocol to ensure reliable message delivery and mitigate congestion, even in the presence of adversarial behavior. Third, we adapt and optimize PBFT and HotStuff to the LEO context, achieving linear message complexity while preserving safety and liveness. Extensive evaluations in a realistic Starlink-based simulation demonstrate that OrbitBFT significantly improves scalability, throughput, and latency compared to its original designs, making it a practical and efficient BFT solution for large-scale satellite networks.

\end{abstract}


\section{Introduction}
\label{sec:introduction}

Low Earth Orbit (LEO) constellations are rapidly emerging as a key component of next-generation integrated space-ground communication infrastructures~\cite{westphal2023leo, reddy2023low, xiao2022leo}. Their inherent advantages, including near-global seamless coverage and reduced transmission delays compared to geostationary satellites, drive this trend~\cite{kodheli2020satellite}. Prominent examples like SpaceX's Starlink, which as of 2026 comprises over 9,000 satellites with plans for 12,000,
organized into multiple orbital planes, exemplify the rapid expansion and foundational role of LEO constellations. Other constellations like Iridium also demonstrate the viability and importance of this technology.
China has submitted two filings for massive satellite networks totaling nearly 200,000 satellites~\cite{jones2026china}, signaling further congestion in near-Earth space.
With advancements in satellite technology and onboard processing capabilities, LEO satellites are evolving from simple communication relays into autonomous network platforms capable of complex distributed coordination and decision-making~\cite{wu2023comprehensive,araguz2018applying}. This evolution enables advanced applications such as collaborative earth observation, real-time onboard data processing, and autonomous mission planning. A fundamental cornerstone for these applications is establishing and maintaining a reliable and consistent state among distributed satellite nodes to enable secure and efficient collaboration.


However, the expansion of LEO constellations introduces severe security and stability challenges. The orbital environment is becoming not only physically congested but also increasingly contested and complex. Collision risks are escalating, evidenced by near-miss incidents between Starlink satellites and the Chinese Space Station~\cite{cnn2021musk}. Furthermore, satellites now face intensifying threats ranging from signal jamming to cyber-hijacking and dogfighting maneuvers~\cite{mccarthy2025china, davenport2025dogfighting}. Under such attacks, a compromised satellite acts as a Byzantine node—an internal adversary that arbitrarily deviates from protocols or broadcasts falsified trajectory data to disrupt the network. Reliance on centralized ground stations is insufficient because uplinks can be severed or jammed in these high-stakes contexts~\cite{salem2026iran}. Consequently, autonomous coordination mechanisms are necessary. For instance, in autonomous collision avoidance, a compromised satellite might broadcast falsified trajectory data to cause collisions. Reliable satellites must reach a consensus on the correct maneuvering plan instantly, without waiting for ground instructions that may be jammed. Similarly, during constellation-wide reconfiguration, satellites must autonomously coordinate their slots to prevent disorder, ensuring resilience even if the central ground control is severed. In these scenarios, Byzantine Fault-Tolerant (BFT) consensus becomes a prerequisite, as it guarantees the constellation reaches a consistent agreement even despite internal subversion, ensuring mission safety.

Despite this urgent need, applying conventional BFT protocols to LEO constellations presents significant difficulties. 
While mature BFT consensus protocols like PBFT~\cite{castro1999practical} and HotStuff~\cite{yin2019hotstuff} have proven effective in the stable, resource-abundant environments of traditional ground-based networks, their inherent design principles clash with the realities of LEO constellations. 
This fundamental mismatch presents a hurdle:

\textit{How to Adapt BFT Consensus Protocols to LEO Constellations with Dynamic but Fragile Topology?}

LEO constellations exhibit a unique and dynamic network topology that remains largely unexplored in existing research on wireless BFT consensus~\cite{zou2024survey, guo2024hierarchical} and Internet routing. Within a LEO constellation, satellites in the same orbital plane maintain stable links, whereas inter-plane connections change rapidly due to orbital motion~\cite{wysocki2021performance}. This dynamic topology introduces challenges. 
First, frequent and complex dynamic routing is essential for inter-plane communication but imposes significant computational overhead~\cite{delportillo2019technical, fink2021constellation}. This burden is further exacerbated by the intensive message broadcasting required for consensus protocols, which strains the limited onboard processing capabilities.
Recent comprehensive reviews indicate that traditional Proof-of-Work (PoW) mechanisms are fundamentally impractical for satellite networks due to their excessive power requirements and inability to handle asynchronous links~\cite{probert2025review}.
Second, the large scale of these networks makes traditional BFT protocols of high message complexity impractical. These factors underscore the critical need for a new consensus architecture designed specifically for the LEO environment.

Another fundamental challenge lies in the sparse connectivity of LEO constellations. Theoretical research has established tight necessary conditions on the underlying graph connectivity required to solve Byzantine consensus in directed networks~\cite{tseng2015fault}. At any given time, a satellite typically connects to only four neighbors~\cite{topl}, which places the network near these critical topological limits. Consequently, communication paths are inherently fragile. A single Byzantine node can disrupt consensus, not through sophisticated attacks, but by selectively dropping or refusing to forward critical messages along key paths. Conventional BFT protocols often assume a well-connected network, where any pair of nodes can communicate reliably via dense P2P overlays, effectively abstracting away routing-level failures. As a result, few BFT protocols account for the constellation's link-layer dynamics. In LEO networks, however, the consensus protocol must incorporate mechanisms that ensure communication reliability in the presence of Byzantine behavior and physical link volatility.

This paper focuses on designing BFT consensus protocols for the challenging LEO satellite network environment. Addressing the unique characteristics of large-scale LEO constellations, we investigate novel BFT protocols optimized for dynamic and sparse topology. Our main contributions are:

\begin{enumerate}
    \item \textbf{Hierarchical BFT Architecture Tailored for LEO Constellations:} We design \protocolname{}, a novel two-stage consensus framework that decomposes the global consensus task into intra-plane and inter-plane phases. By exploiting the topological stability within orbital planes and limiting inter-plane coordination to a small committee, OrbitBFT significantly reduces bandwidth overhead and improves scalability in large LEO constellations.
    
    \item \textbf{Byzantine-Resilient Bypass Mechanism:} To counteract path fragility and malicious behavior, we propose a Byzantine-resilient Bypass Mechanism that dynamically reroutes messages via adjacent planes, and a congestion-aware hop-by-hop transmission protocol that ensures efficient and reliable communication within intra-plane rings.
    
    \item \textbf{Protocol-Level Adaptation and Empirical Validation of BFT:} We tailor PBFT and HotStuff to the LEO constellation's topology. These adaptations achieve linear message complexity, mitigate congestion, and preserve safety/liveness under constrained connectivity. Extensive simulations using realistic Starlink parameters validate the superiority of our optimized protocols over their native counterparts in terms of scalability, throughput, and latency.
\end{enumerate}

\begin{table*}[!t]
\centering
\begin{threeparttable}
\caption{A Comparison of BFT Consensus on Earth or in Space}
\label{table:comparision}
\tabcolsep=0.3cm
\begin{tabular}{l|c|c|c|c|c}
\toprule[1pt]
& \textbf{Target Network} & \textbf{Architecture} & \textbf{Scalability} & \textbf{\makecell[c]{Message\\Complexity}} & \textbf{\makecell[c]{Inter-plane\\Link Utilization$^\ddagger$}} \\
\midrule[0.5pt]
PBFT~\cite{castro1999practical}  & \multirow{3}{*}{\makecell[c]{Wired Network\\(Earth)}} & Distributed & \Circle & $O(N^2)$ & $O(k^2n^2)$ \\
HotStuff~\cite{yin2019hotstuff}  & & Distributed & \Circle & $O(N)$ & $O(k^2n^2)$ \\
Dynamic BFT~\cite{DynamicBFT}  & & Distributed  & \LEFTcircle & $O(N^2)$ & $O(k^2n^2)$ \\
\midrule
BLOWN~\cite{BLOWN} & \multirow{3}{*}{\makecell[c]{Wireless Network\\(Earth)}} & Distributed & \LEFTcircle & $O(N)$ & $O(kn)$ \\
TinyBFT~\cite{bohm2024tinybft}  & & Distributed & \Circle & $O(N^2)$ & $O(k^2n^2)$ \\
ReduceCatch~\cite{liu2025partially} & & Distributed & \LEFTcircle & $O(N\log N)$ & $O(k^2n^2)$ \\
\midrule
HPBFT~\cite{wang2024enabling}  & \multirow{2}{*}{Satellite Network} & Distributed & \LEFTcircle & $O(N\log N)$ & $O(kn)$ \\
SatBFT~\cite{zhang2025satbft}  & & Centralized & \Circle & $O(N^2)$ & $O(1)^*$ \\
\midrule
\textbf{OrbitBFT$^\dagger$} & Satellite Network & Distributed & \CIRCLE & $O(N)$ & $O(k)$ \\
\bottomrule[1pt]
\end{tabular}
\begin{tablenotes}

    \item[$\dagger$] OrbitBFT assumes a local fault tolerance threshold ($n \ge 3f_i+1$) per plane, whereas traditional BFT assumes a global threshold ($N \ge 3F+1$).
    \item[$\ddagger$] Frequency of inter-plane link usage per consensus instance. We use $N$ to denote the total number of satellites in the LEO constellation, $k$ to represent the number of orbital planes, and $n$ to represent the number of nodes per orbital plane, such that $N = k \times n$. The factored form is retained to highlight the specific impact of the orbital topology (planes vs. nodes) on link utilization.
    \item[$*$] SatBFT has $O(1)$ inter-plane utilization as it relies on high-orbit GEO satellites for global communication, thus not scaling with the number of LEO satellites.
\end{tablenotes}
\end{threeparttable}
\end{table*}


\section{Related Work} 
\label{sec:related_work} 

\subsection{BFT Consensus on Earth}
BFT consensus protocols on Earth can be categorized into those designed for wired and wireless networks. In wired networks, consensus protocols rely on a reliable link between any two nodes. PBFT~\cite{castro1999practical} is a groundbreaking BFT consensus that has influenced numerous subsequent BFT consensus protocols, including Tendermint~\cite{Tendermint} and HotStuff~\cite{yin2019hotstuff}. Notably, HotStuff is the first partially synchronous BFT consensus protocol to achieve both linearity and responsiveness. Subsequent research has acknowledged that wired network links may not always be ideal, and consensus nodes may dynamically enter and exit. To address these challenges, Duan \textit{et al.}~\cite{DynamicBFT} propose the Dynamic BFT protocol, a framework designed to manage dynamic node participation and departure, particularly in permissioned blockchain contexts.

In wireless networks, wireless channels are susceptible to environmental interference, preventing the establishment of reliable links between any two nodes, unlike wired networks. BFT consensus protocols in wireless networks employ various techniques, including Trusted Execution Environment (TEE)~\cite{TEE-1}, spanner~\cite{wChain}, proof-of-channel~\cite{BLOWN}, carrier sensing, and random timer, to handle collisions~\cite{802.11_PBFT}. TinyBFT~\cite{bohm2024tinybft} reduces PBFT's storage footprint, enabling its deployment on resource-constrained ESP32 devices. Furthermore, Liu \textit{et al.}~\cite{liu2024partially} propose ReduceCatch, a wireless communication protocol designed to facilitate the adaptation of partially synchronous BFT consensus protocols from wired to wireless environments.

\subsection{BFT Consensus in Space}
Research on BFT consensus in space is still in its nascent stage, evolving from specific application scenarios to fundamental infrastructure optimizations. Specific studies have focused on tailoring blockchain for specific satellite applications. For instance, a fault-tolerant consensus mechanism~\cite{mollakhani2025fault} was proposed for inter-operator spectrum sharing, enabling approximate agreement on noisy measurements among operators. Similarly, a reputation-based blockchain framework~\cite{shen2025privacy} was designed to secure federated learning and pricing in LEO IoT. Beyond specific applications, efforts have been made to optimize the underlying blockchain scalability. The Age-critical Blockchain Sharding (ABS) scheme~\cite{wang2022age} partitions the network based on geographic domains to optimize information freshness. However, a common limitation of these works is their reliance on ground stations or static terrestrial nodes for consensus execution. By treating satellites merely as relays or data collectors, these approaches bypass the complex challenges of on-orbit geometric dynamics and link instability.

To enable consensus directly on dynamic LEO infrastructures, researchers have attempted to adapt network architectures. To enable consensus directly on dynamic LEO infrastructures, researchers have proposed various architectural adaptations, primarily utilizing hierarchical structures. Hierarchical PBFT (HPBFT)~\cite{wang2024enabling} simplifies the topology by segmenting the constellation into clusters to reduce message complexity. Similarly, SatBFT~\cite{zhang2025satbft} adopts a multi-layer strategy, leveraging GEO satellites as stable relays to manage consensus messages. While these hierarchical designs improve scalability, they often introduce dependencies on specific cluster heads or external high-orbit infrastructures.

\subsection{Limitations of Existing Works}
As illustrated in Table~\ref{table:comparision}, current approaches have contributed to BFT consensus but still face several limitations when adapted to LEO constellations: 

1) Most BFT consensus protocols on Earth, such as PBFT and HotStuff, need reliable links between any two nodes. This requires establishing $O(N^2)$ reliable links and extensive inter-plane routing, which would consume a large volume of satellite resources and limit their scalability in LEO constellations. While BFT consensus~\cite{BLOWN, bohm2024tinybft, liu2025partially} in dynamic and wireless networks does not require reliable links between any two nodes, if deployed in LEO constellations, they would still require frequent invocation of inter-plane routing protocols, thus limiting scalability. 

2) Existing in-space consensus algorithms have made some improvements for scalability within LEO constellations, yet they still have limitations in large-scale LEO constellations.

For instance, HPBFT~\cite{wang2024enabling} segments the constellation into clusters to reduce complexity. However, nodes within the same consensus cluster often span multiple adjacent orbital planes. This creates a dependency on unstable inter-plane links for intra-cluster consensus, maintaining high routing costs.
SatBFT~\cite{zhang2025satbft} utilizes GEO satellites as relays to facilitate LEO consensus. While this reduces LEO-to-LEO message complexity, it introduces a reliance on high-orbit infrastructure. This hierarchical dependency creates a central point of failure and bottleneck at the GEO relay nodes, contradicting the decentralized requirements of robust mega-constellations.
Therefore, existing work has not been able to provide good scalability in large-scale LEO constellations. The main reason for this is excessive inter-plane routing calls, which consume lots of resources and, in turn, limit scalability.

\section{Preliminaries and Models}
\label{sec:system_model}

\textbf{Partially Synchronous BFT Consensus.}
Partially synchronous BFT consensus protocols use the partially synchronous network model to overcome the FLP impossibility~\cite{FLP}. They guarantee safety and liveness even with up to $f$ Byzantine nodes, where $N=3f+1$ and $N$ is the network size. The goal of this paper is to efficiently adapt partially synchronous BFT consensus to satellite networks, with the core idea being cross-layer optimization from the network perspective. Therefore, the underlying communication patterns of partially synchronous BFT consensus are crucial. For completeness, we select two well-established protocols, PBFT~\cite{castro1999practical} and HotStuff~\cite{yin2019hotstuff}, as they cover all communication patterns of partially synchronous BFT consensus: 1-to-N (one-to-all), N-to-1 (all-to-one), and N-to-N (all-to-all).

\begin{figure}[htbp]
\centering
\begin{subfigure}[b]{0.48\columnwidth}
    \includegraphics[width=\textwidth]{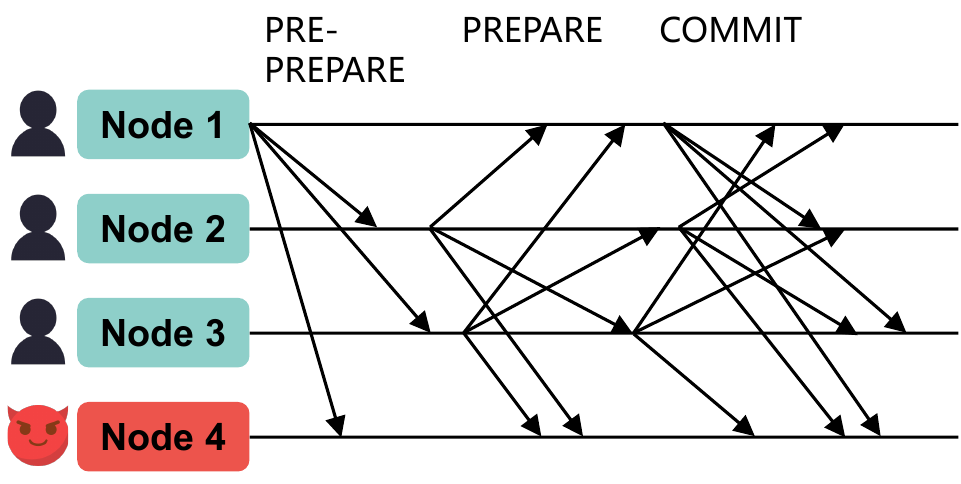}
    \caption{PBFT}
    \label{fig:PBFT}
\end{subfigure}
\hfill
\begin{subfigure}[b]{0.48\columnwidth}
    \includegraphics[width=\textwidth]{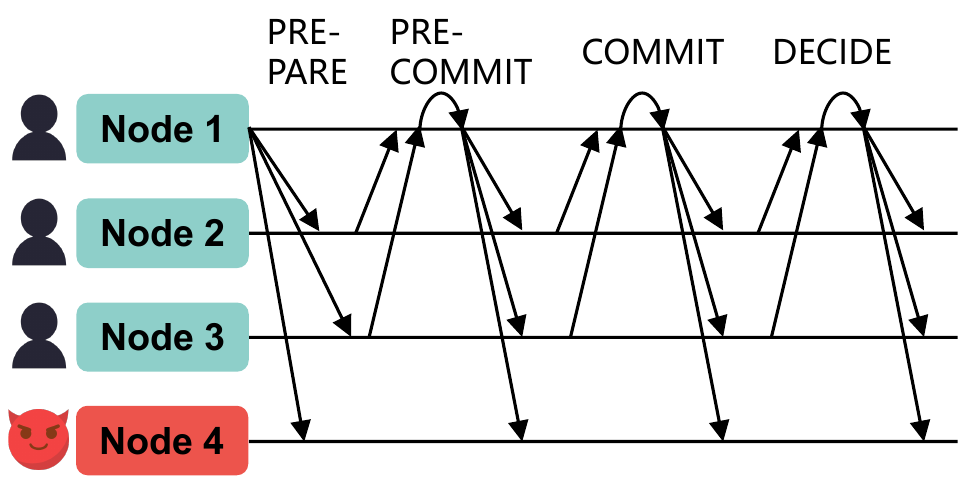}
    \caption{HotStuff}
    \label{fig:HotStuff}
\end{subfigure}
\caption{Partially synchronous BFT consensus}
\label{fig:consensus}
\end{figure}

In this paper, we consider the normal case operation of partially synchronous BFT consensus. As shown in Fig.~\ref{fig:PBFT}, PBFT consists of three phases. The leader first uses a 1-to-N communication to send its proposal to other followers. Then, two phases of N-to-N communication for voting ultimately determine whether the proposal is committed. As shown in Fig.~\ref{fig:HotStuff}, HotStuff consists of four phases. Similar to PBFT, the leader first sends the proposal to other followers via a 1-to-N communication. Unlike PBFT, HotStuff uses threshold signatures in the subsequent three phases. In each of these subsequent phases, the leader first collects votes from followers and aggregates them into a threshold signature, which is then broadcast to other followers. Therefore, the latter three phases utilize both 1-to-N and N-to-1 communication.

\textbf{Network Model.}
\label{sec:network_model}
We consider a large-scale LEO constellation, defined by a set $V$ of $N$ satellites. These satellites are configured across $k$ near-circular orbital planes, with each plane consisting of $n$ satellites ($N = k \times n$). We assume that the satellites within the same orbital plane are distributed uniformly. 
The LEO constellation can be modeled as a spherical mesh in 3D coordinate space.
Every satellite is assigned a distinct identifier and is initialized with knowledge of both the plane size and the total constellation size. Furthermore, we adopt a partially synchronous network model~\cite{dwork1988consensus}. Each orbital plane alternates between asynchronous and synchronous states. During an asynchronous state, a Global Stabilization Time (GST) is established, after which the network becomes synchronous. This assumption of partial synchrony is justified for an orbital plane due to the inherent susceptibility of the inter-satellite link (ISL) to cosmic radiation and other environmental influences. 

\begin{figure}[!htb]
    \centering
    \includegraphics[width=0.9\linewidth]{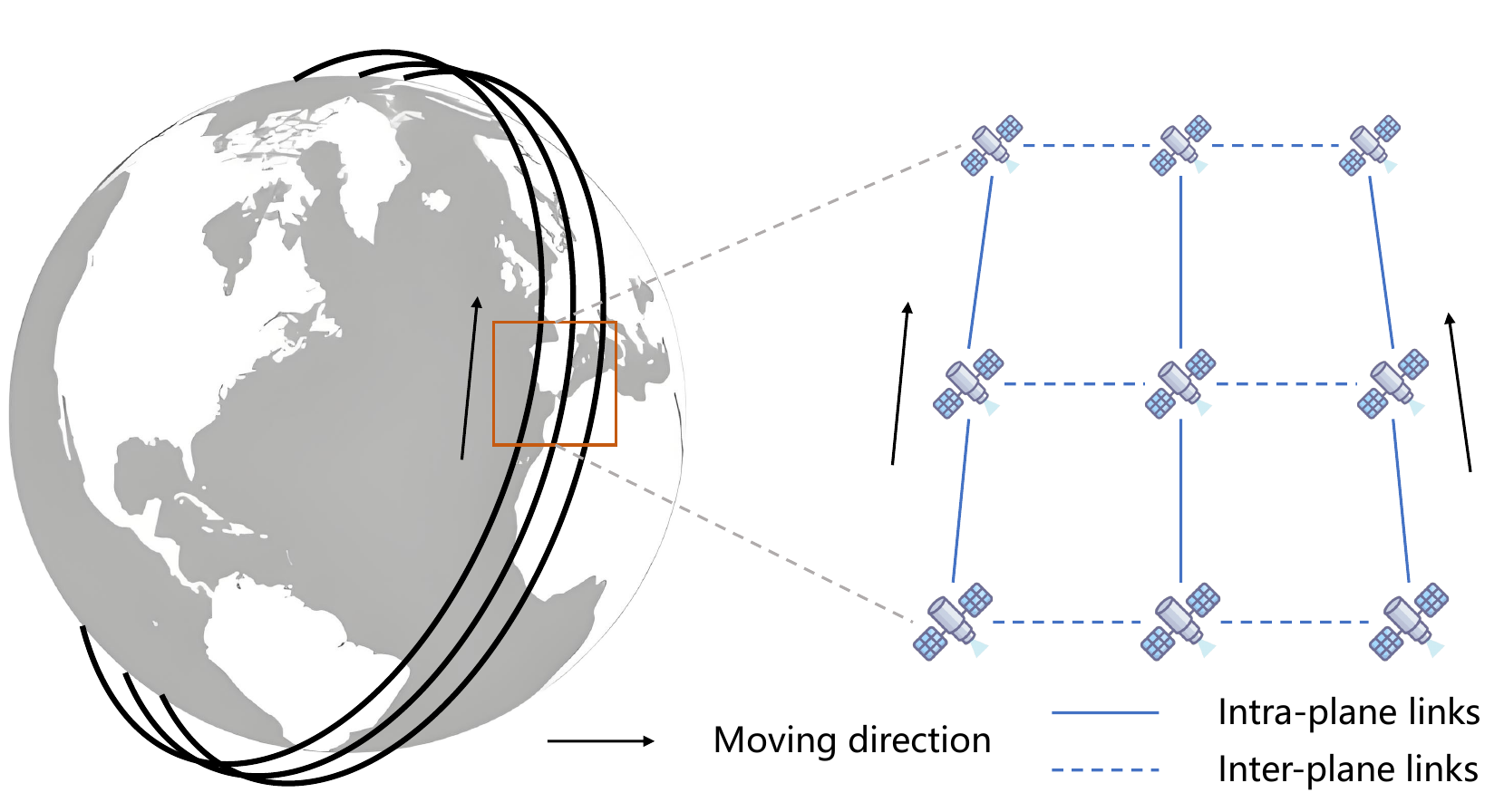}
    \caption{LEO constellation. The black lines on the left side of the figure represent the satellite's orbits.}
    \label{fig:LEO}
    
\end{figure}

As shown in Fig.~\ref{fig:LEO}, our LEO constellation uses two types of ISLs with different properties. Intra-plane ISLs connect nearby satellites in the same orbital plane; these links are very stable and form $k$ unchanging ring-like structures. In contrast, Inter-plane ISLs connect satellites in different orbital planes. These connections are dynamic, constantly forming and breaking as satellites move to or leave other satellites. Each satellite in a plane maintains at most four ISLs: two intra-plane links with its neighbors in the same plane and two inter-plane links with satellites in adjacent orbital planes. Our connectivity model aligns with practical designs of modern LEO constellations and simplifies the network model for BFT consensus protocol analysis~\cite{chaudhry2022temporary}.

\textbf{Adversary Model.}
Among a total of $N$ nodes, we assume there are at most $F$ Byzantine nodes, which can behave arbitrarily by deviating from the protocol, such as sending falsified messages, withholding messages, or colluding with other malicious nodes. The nodes are partitioned into $k$ consensus groups, each with $n$ nodes, where each group has at most $f_i$ Byzantine nodes, with $n=3f_i+1$. Therefore, $F=f_1+f_2+...+f_k$. 
We assume that Byzantine nodes do not form a physically isolated closure. A closed-off group of nodes would be cut off from the network, making it impossible for any protocol to ensure reliable communication. We assume such permanent encirclement is probabilistically negligible due to the high-velocity relative motion of satellites, which naturally resolves physical isolation over time.

\section{\protocolname{}: BFT Consensus for LEO Constellations}
\label{sec:architecture}

\subsection{Motivation}
\label{sec:motivation}

Our work is fundamentally motivated by a crucial finding within LEO satellite networks: while the global network topology exhibits high dynamism due to the motion of satellites and the temporary nature of ISLs~\cite{topl}, complicating real-time dynamic routing~\cite{han2022time, zhang2024adaptive}, satellites maintain stable relative positions, forming a ring topology within each orbital plane. These satellites in same orbital plane are interconnected via persistent and highly reliable ISLs, enabling simple and stable routing. In contrast to the static connectivity of traditional Internet or fully randomized topologies, this local stability where global dynamism coexists with invariant local structures, is a powerful and underexploited feature.

Recognizing the inherent stability in LEO networks, we propose \protocolname{}, a two-stage consensus protocol. Our core idea is to simplify the complex global consensus problem by decomposing it into two distinct subproblems: intra-plane consensus and inter-plane consensus. Intra-plane consensus ensures consensus among satellites within the same orbital plane. Inter-plane consensus takes responsibility for globally ordering the proposals generated by these intra-plane consensus processes. By localizing consensus efforts within planes before aggregating them globally, \protocolname{} reduces the communication burden and computational complexity typically associated with achieving network-wide agreement in large-scale satellite systems.

\subsection{Two-Stage \protocolname{}}

\protocolname{} partitions the entire constellation network into multiple parallel consensus groups based on their physical orbital planes. Each orbital plane constitutes an intra-plane consensus group, coordinated by a deterministically elected plane leader (PL). For the inter-plane stage, a global leader (GL) is elected to finalize the global agreement. Nodes in the network thus assume one of three roles: a regular node within a plane, known as a Replica; a PL; and a GL, who also serves as the PL for its own plane. The complete consensus process in \protocolname{} unfolds in two core stages, as shown in algorithm~\ref{alg:orbitbft_consensus}.

\textbf{Intra-plane consensus.}
During this stage, each PL kicks off a standard BFT consensus process within its designated group. If a group successfully reaches a consensus, it gets a proposal that has been certified by the signatures of a quorum of nodes in its orbital plane.

\textbf{Inter-plane consensus.}
After successfully generating a proposal in the first stage, a GL is elected. All PLs then send their proposals to the elected GL. The GL then combines all the collected proposals into an ordered sequence using a deterministic function. Finally, the GL initiates a final BFT consensus instance on this combined sequence committee of PLs (deterministically selected via a pseudo-random function seeded by the ViewID, e.g., of size $k$). If this global consensus succeeds, all PLs broadcast the final sequence to the replicas in their respective groups, committing the state. If it does not succeed, the system moves to a new consensus epoch, which includes re-electing a leader.

Our hierarchical architecture creates a trust boundary where Replicas rely on their PL for inter-plane communication. While a Byzantine PL could temporarily impact local consistency, this is a necessary trade-off for linear scalability ($O(N)$). The protocol mitigates this risk through the View Change mechanism and Rotational Leadership, ensuring that any such malicious influence is transient and confined to a single plane, preventing global system paralysis.

\begin{algorithm}[!t]
\caption{Two-Stage Consensus Protocol in \protocolname{}}
\label{alg:orbitbft_consensus}
\begin{algorithmic}[1]
\State \textbf{Stage 1: Intra-plane Consensus}
\For {each $PL$}
    \State $PL$ initiates BFT consensus
    \If {$PL$ collects a quorum of signatures}
        \State partial\_certificate $\gets$ \Call{GeneratePartialCertificate}{}
    \EndIf
\EndFor

\State 
\State \textbf{Stage 2: Inter-plane Consensus}
\State  $GL \gets$ \Call{SelectGlobalLeader}{}
\State $PLs$ send proposals to $GL$ \Comment{Parallel \& Async}
\State $GL$ aggregates proposals (wait for quorum or timeout)
\State $committee \gets$ \Call{Select committee}{$PLs, N_c$}
\State $GL$ initiates BFT consensus on sequence with 
\State $committee$ broadcast result to all $PLs$
\textit{committee}
\If {consensus with $committee$ succeeds}
    \For {each $PL$}
        \State $PL$ broadcasts final sequence to replicas
    \EndFor
\Else
    \State Proceed to new consensus round
\EndIf
\end{algorithmic}
\end{algorithm}

\subsection{Intra-Plane Consensus}

Traditional BFT consensus protocols, when directly applied to LEO constellations, encounter security and efficiency challenges. BFT consensus typically relies on communication patterns such as 1-to-N, N-to-1, and N-to-N. On the internet, particularly in peer-to-peer (P2P) networks, the high degree of network connectivity allows nodes to freely broadcast proposals or votes to a large number of peers through gossiping mechanisms.
In stark contrast, LEO constellations possess a more rigid, mesh-like topology, where a satellite is typically connected to only four adjacent satellites at any given time (two within the same orbital plane and two in adjacent planes). This constrained connectivity presents a critical problem: if a Byzantine node is adjacent to an honest node and attempts to disrupt the ISL, the honest node has only three alternative paths to bypass the malicious entity. Furthermore, this limited connectivity inherently exacerbates network congestion, as each satellite has only four available links for message transmission, restricting its capacity to handle large volumes of consensus-related traffic.

\subsubsection{Bypass Mechanism}
\label{sec:bypass}

In a ring topology, a Byzantine node can disrupt message delivery along a communication path. To address this, we introduce a Bypass Mechanism. This mechanism is designed to guarantee reliable message delivery within our spherical mesh topology, even when Byzantine nodes are present. Our protocol's primary mode of communication is within the same orbital plane. To illustrate the Bypass Mechanism, we'll use a basic example: a node $X$ sending a message to another node $Y$ within the same plane. This instance illustrates the general mechanism. The mechanism is illustrated in three procedures. 

\textbf{Acknowledgment Confirmation.} Our mechanism first requires the receiving node to send an acknowledgment (ACK) back to the sender to confirm delivery. When node $X$ sends a message to node $Y$, it starts a timer and expects an ACK to arrive within a predefined time window. If this ACK is not received, $X$ cannot confirm that the message was successfully delivered.

\textbf{Byzantine Node Detection and Bypass.} When node $X$ does not receive an expected ACK from node $Y$, it suspects a Byzantine fault along the path but can not determine the exact location. This triggers a bypass. $X$ sends the message to its neighbors in both adjacent orbital planes. For simplicity, let us focus on the path through a node we call $X'$, as shown in Fig.~\ref{fig:toy_example}. This reactive behavior is a core rule of the protocol. If any node responsible for forwarding the message fails to receive an ACK, it concludes that its current path is compromised and initiates its own bypass by contacting a neighbor in an adjacent plane. If this bypass fails (i.e., node $X$ does not get a response from $X'$), $X$ contacts its predecessor, $X-1$. The predecessor then starts a wider bypass designed to circumvent the unresponsive neighbor's area. This process of requesting help from a predecessor continues recursively. However, to ensure bounded latency, if the backtracking depth reaches the plane size $n$ (implying a traversal of the entire ring without success), the node concludes that the local topology violates the connectivity assumption (i.e., a closure has formed) and triggers a view change.

\begin{figure}[!t]
\centering
\begin{subfigure}[b]{0.48\columnwidth}
    \includegraphics[width=\textwidth]{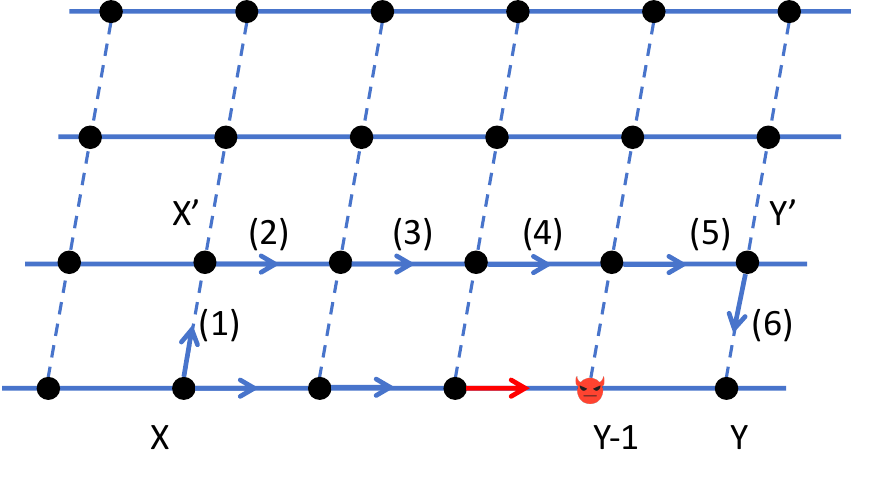}
    \caption{A toy example}
    \label{fig:toy_example}
\end{subfigure}
\hfill
\begin{subfigure}[b]{0.48\columnwidth}
    \includegraphics[width=\textwidth]{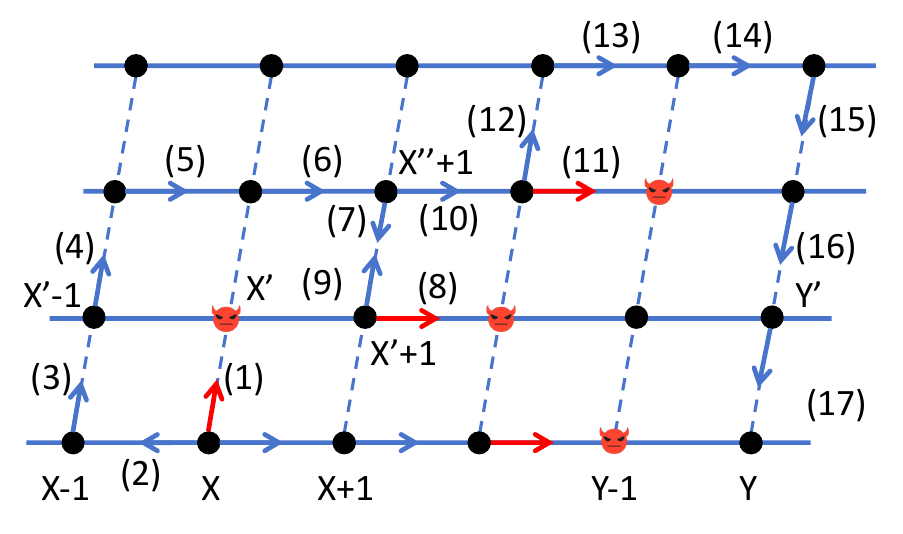}
    \caption{A complex example}
    \label{fig:complex_example}
\end{subfigure}
\caption{When $X$ fails to receive an ACK from node $Y$, messages are passed in numerical order, with red indicating failed transmissions.}
\label{fig:example}
\end{figure}

\textbf{Path Segmentation and Responsibility.} The act of initiating a bypass by having node $X$ find a cross-plane link to send its message naturally divides the path into several segments. The example shown in Fig.~\ref{fig:toy_example} has three segments: $X \to X'$, $X' \to Y'$, and $Y' \to Y$. The node at the start of each segment is responsible for delivering the message to the start of the next segment. This responsibility is only fulfilled upon receiving an ACK from the original target $Y$. If a segment's starting node is unable to receive an ACK from $Y$, it will recursively invoke the Bypass Mechanism to find a new path that circumvents the Byzantine node.

Through these procedures, a complete Bypass Mechanism is constructed, ensuring that a reliable path can always be found. This process may result in multiple concurrent bypass attempts. The first attempt that successfully returns the final ACK from $Y$ completes the message passing. Fig.~\ref{fig:complex_example} illustrates a more complex bypass condition where there are multiple Byzantine nodes along the path besides $Y$. After $X$ initiates the first bypass, it fails to transmit a message because $X'$ is a Byzantine node. It then requests assistance from its predecessor, $X-1$, which starts the second bypass. This process may repeat until, in this example, the message reaches $Y$.

\subsubsection{Hop-by-hop Transmission}
To address network congestion on an orbit, we implemented a hop-by-hop transmission protocol instead of traditional broadcast mechanisms. The protocol is designed for the specific topology of an orbital plane where nodes form a ring. Since consensus messages are identical for all participating nodes, a PL initiates a message once, sending it in both directions along the orbital plane. The message is forwarded and processed by each node until it reaches the node symmetrically opposite the PL on the ring. 
To protect against Byzantine nodes that might intercept messages or collude to send acknowledgments without propagating the message, each node along the path is required to send the message to its next $f_i+1$ successors and receive an ACK. This ensures that at least one honest node receives and continues to propagate the message, completing the ``relay race'' and guaranteeing message propagation across the orbit.

Based on this hop-by-hop transmission, BFT protocols such as PBFT and HotStuff can be efficiently adapted to LEO constellations. For PBFT, the PL broadcasts a \texttt{PRE-\texttt{PREPARE}} message hop-by-hop. Replica nodes, upon receiving it, proceed with \texttt{PREPARE} and \texttt{COMMIT} phases via hop-by-hop transmission, collecting $2f+1$ \texttt{COMMIT} messages to complete intra-plane consensus. For HotStuff, the PL broadcasts proposals for \texttt{PREPARE}, \texttt{PRE-COMMIT}, and \texttt{COMMIT} phases hop-by-hop. Other nodes validate proposals and return signed votes, collecting $2f+1$ \texttt{COMMIT} votes to form a Quorum Certificate, completing intra-plane consensus.

\subsubsection{Congestion-Aware Flow Control} 
Although the hop-by-hop mechanism serializes message transmission to avoid collisions, an unbounded number of concurrent consensus instances can still accumulate to exceed the link capacity. To ensure system stability under high loads, \protocolname{} incorporates a sliding window mechanism. Each PL maintains a congestion window of size $W$, limiting the number of active uncommitted proposals in the pipeline. New consensus instances are only initiated when the number of pending instances falls below $W$.

\subsection{Inter-Plane Consensus}
\label{sec:inter_consensus}

The global consensus procedure builds upon the reliable inter-PL communication mechanism. 
While the intra-plane stage generates local certificates, the inter-plane stage is responsible for ordering these batches globally and ensuring state consistency across the constellation.

\subsubsection{Global Leader Election and Committee}
To coordinate the inter-plane phase, a GL is elected deterministically from the set of PLs. 
The election follows a rotation-based schedule coupled with the view number.
Considering the scalability, the GL does not necessarily run BFT among all $k$ PLs if $k$ is extremely large. 
Instead, a committee of size $N_c$ (where $N_c \le k$) can be selected. 
To ensure scalability independent of the constellation size ($k$), we set the committee size to a small constant denoted as $N_c$ (where $N_c \ll k$). The committee members are deterministically selected from the $k$ PLs based on the current view number.
We define $f_g$ as the maximum number of Byzantine nodes tolerated within this inter-plane committee, satisfying $k \ge 3f_g + 1$.

\subsubsection{Global Aggregation and Consensus}
The inter-plane consensus operates as follows:
Upon election, the GL waits for a predefined time window $\Delta_{collect}$ to receive proposals and partial certificates from all PLs. 
To ensure liveness, the GL does not wait for all $k$ PLs. 
It proceeds as long as it receives valid inputs from at least $2f_g+1$ PLs, or simply packs whatever it has received when $\Delta_{collect}$ expires.
The GL aggregates the collected valid proposals into a global block and initiates a BFT consensus instance (e.g., PBFT or HotStuff) within the committee of PLs. 
Communication during this phase utilizes the reliable inter-plane routing described in Section~\ref{sec:bypass}.
Once the global consensus is reached (i.e., a PL receives $2f_g+1$ commit signatures), the PL considers the global block committed. 
Finally, each PL broadcasts this finalized global block down to all Replicas in its own orbital plane via the intra-plane ring, ensuring global state consistency across all satellites.

    
    

\subsection{Leader Election and View Change}
\label{sec:view_change}

To ensure liveness and handle Byzantine failures—whether they are faulty nodes or leaders—\protocolname{} adopts a deterministic view-change protocol inspired by PBFT~\cite{castro1999practical}.

\subsubsection{Deterministic Leader Selection}
Time is divided into views, denoted by an integer $v$. In each view, a unique leader is determined deterministically, avoiding the overhead of complex election algorithms.
For an orbital plane with $n$ nodes indexed $0$ to $n-1$, the PL for view $v$ is the node with ID $p = v \pmod n$. 
Similarly, for the inter-plane committee with $k$ PLs, the GL is the PL with index $g = v \pmod k$.
This round-robin rotation ensures that if a leader is a Byzantine node, the protocol will eventually rotate to an honest leader within at most $f_i+1$ view changes.




\subsubsection{Exception Handling via View Change}
A view change is triggered to advance $v$ to $v+1$ when progress stalls. This guarantees the liveness of the system.
If a node (Replica or PL) sends a transaction proposal but does not receive a consensus result within a timeout $T_{out}$, it broadcasts a \texttt{VIEW-CHANGE} message. Once a node collects $2f+1$ \texttt{VIEW-CHANGE} messages, it enters the new view $v+1$.
To address the potential for unbounded latency in the recursive Bypass Mechanism, we enforce a strict limit on recursion depth (e.g., equal to the plane size $n$). If a valid path cannot be found within this limit—implying a network partition or excessive Byzantine nodes forming a closure—the node immediately aborts routing and triggers a \texttt{VIEW-CHANGE}. This ensures that routing failures do not cause indefinite waiting, satisfying the liveness requirement.



\section{Analysis}
\label{sec:analysis}

This section theoretically proves the correctness of the \protocolname{} protocol, demonstrating that it satisfies the essential properties of safety and liveness required for distributed consensus. We also provide a complexity analysis to quantify its performance advantages.

\subsection{Safety and Liveness of the \protocolname{} Protocol}

\begin{theorem}[Safety and Liveness of \protocolname{}]
The \protocolname{} protocol ensures both safety and liveness. Safety guarantees that all honest nodes agree on the same sequence of valid transactions. Liveness guarantees that any valid transaction proposed by an honest node is eventually committed by all honest nodes. 
\end{theorem}
\begin{proof}
These properties hold given that the underlying BFT protocol is also safe and live. The Safety and Liveness of \protocolname{} hinges on the reliability of its communication layer, which must ensure that messages between honest nodes are not lost or corrupted. We establish the reliability of our communication primitive with the following lemma.

\begin{lemma}[Reliability of the Bypass Mechanism]
\label{lem:bypass_reliability}
The Bypass Mechanism guarantees reliable message delivery between any two honest nodes.
\end{lemma}
\begin{proof}
The proof relies on two fundamental properties from our adversary model: (1) each orbital plane satisfies $n = 3f_i + 1$, and (2) Byzantine nodes do not form a physically isolated closure. To prove reliability, we must show that the mechanism can overcome all possible disruptive behaviors from Byzantine nodes. These behaviors can be categorized into two fundamental types. First, a Byzantine node can directly disrupt a path by refusing to forward a message or its acknowledgment, causing a timeout. Second, a Byzantine node can act deceptively by correctly forwarding a message within a local segment, only to intercept the final acknowledgment from the ultimate destination. The Bypass Mechanism addresses both failure types.

If a message's direct path is disrupted, a node will first attempt a delivery within its own orbital plane. If this fails, it initiates a bypass to an adjacent plane, or its predecessor for a wider bypass. This deterministic search ensures that as long as an honest node is not completely isolated, a viable path to another honest node will eventually be quickly found. To prevent deceptive forwarding, a node that initiates a path segment is not considered finished until it receives an end-to-end ACK from the final destination. By combining a deterministic path exploration strategy with an end-to-end confirmation logic, the Bypass Mechanism systematically addresses all possible ways a Byzantine node can disrupt communication, thus ensuring reliable message delivery.
\end{proof}

With the reliability of the communication channel established by Lemma~\ref{lem:bypass_reliability}, we can now prove the two properties for the overall protocol. (1) \textit{Safety:} With Lemma~\ref{lem:bypass_reliability} ensuring a reliable communication channel, both the intra-plane and inter-plane consensus stages operate without the loss or corruption of messages between honest parties. Since the underlying BFT protocols are assumed to be safe, they will function correctly on this reliable foundation, thus ensuring the overall safety of the \protocolname{} protocol. (2) \textit{Liveness:} The protocol's liveness depends on two conditions: reliable message delivery and a mechanism to handle faulty leaders. The first condition is guaranteed by Lemma~\ref{lem:bypass_reliability}, which ensures that messages from honest nodes are not indefinitely lost. The second condition is handled by the standard view-change mechanism inherent in BFT protocols, which replaces faulty leaders. Since messages are not lost and faulty leaders are eventually replaced, the protocol will continue to make progress, thus ensuring it is live.
\end{proof}

\subsection{Complexity Analysis}

In this section, we analyze the message complexity of \protocolname{}, defined as the number of point-to-point message transmissions required for one successful consensus instance in a no-fault scenario. Let $N$ be the total number of nodes, $k$ be the number of orbital planes, and $n = N/k$ be the number of nodes per plane.

\begin{theorem}[Message Complexity of \protocolname{}]
The message complexity of \protocolname{} is $O(N)$ when using either PBFT or HotStuff.
\end{theorem}
\begin{proof}
The total complexity is the sum of its two stages. Let $C_{BFT}(m)$ be the message complexity of the underlying BFT protocol for a group of $m$ nodes.

First, we analyze the complexity of a single intra-plane consensus instance. Our protocol uses a hop-by-hop broadcast model where each message transmission requires an acknowledgment (ACK). Therefore, a broadcast from one node to all others on the ring requires $O(n)$ messages for the proposal and $O(n)$ for the ACKs, for a total of $O(n)$.
\begin{itemize}
    \item When using PBFT, the \texttt{PRE-\texttt{PREPARE}} phase is one broadcast from the leader, costing $O(n)$. The \texttt{PREPARE} and \texttt{COMMIT} phases each involve all $n$ nodes broadcasting to each other, resulting in $n \times O(n) = O(n^2)$ complexity for each phase. Thus, the total complexity for one plane is $C_{BFT}(n) = O(n^2)$.
    \item When using HotStuff, communication is leader-centric. Each of the four phases typically involves a broadcast from the leader ($O(n)$) followed by replies from all nodes ($O(n)$). Thus, the total complexity for one plane is $C_{BFT}(n) = O(n)$.
\end{itemize}

In Stage 1 (Intra-plane Consensus), $k$ such instances run in parallel, leading to a total complexity of $k \cdot C_{BFT}(n)$.

In Stage 2 (Inter-plane Consensus), the process involves two parts. First, all $k$ PLs send their proposals to the GL, which has a complexity of $O(k)$. Second, the GL initiates a global consensus within a deterministically selected committee of PLs of size $N_c$. Since $N_c$ is a pre-configured small constant independent of $N$, the complexity of this step, $C_{\text{BFT}}(N_c)$, is $O(1)$.

The total complexity of \protocolname{} is the sum of these parts: $k \cdot C_{\text{BFT}}(n) + O(k) + C_{\text{BFT}}(N_c)$.
\begin{itemize}
    \item For PBFT, this sum is $k \cdot O(n^2) + O(k) + O(1)$. Since the plane size $n$ is a constant system parameter, $O(n^2)$ is effectively constant relative to $N$. Thus, the total complexity scales as $O(k)$. Since $k = N/n$, we have $O(k) = O(N)$.
    \item For HotStuff, this sum is $k \cdot O(n) + O(k) + O(1)$. Similarly, treating $n$ as a constant, this simplifies to $O(k) = O(N)$.
\end{itemize}
In both cases, the total complexity is $O(N)$. This completes the proof.
\end{proof}

While message complexity measures the total number of messages, we also note the protocol's efficiency at the physical link level. This communication model, designed for the ring topology, ensures that for any single broadcast event, a message traverses a given link a constant number of times, resulting in an $O(1)$ load on each link. In contrast, a naive broadcast where a central node must send individual messages to all other nodes would place an $O(n)$ load on the links directly connected to the sender. Our approach therefore prevents the formation of communication bottlenecks and fundamentally reduces the bandwidth consumption on each physical link.

This efficiency on stable intra-plane links is complemented by the protocol's optimization of the more critical inter-plane links.

\begin{theorem}[Inter-Plane Link Utilization]
The number of messages traversing inter-plane links during one consensus instance is $O(k)$.
\end{theorem}
\begin{proof}
As established in the previous analysis, all inter-plane traffic occurs in Stage 2. It comprises two components: the $O(k)$ messages from all PLs sending their proposals to the GL, and the $O(1)$ messages from the BFT consensus within the constant-sized committee ($N_c$). The total utilization is therefore $O(k) + O(1) = O(k)$.
\end{proof}

This result demonstrates that our architecture confines the communication load on the dynamic and complex inter-plane network to scale linearly with the number of planes, thereby reducing the overhead caused by dynamic routing across the entire network.

\subsection{Robustness and Recovery Latency}
\label{sec:robustness}

While Section~\ref{sec:evaluation} focuses on steady-state performance, the system's robustness against leader failures is guaranteed by the deterministic View Change mechanism inherited from the underlying BFT protocols. Since the protocol behavior is state-deterministic, the recovery latency ($T_{recovery}$) can be analytically bounded.

The recovery latency consists of two components: the failure detection time and the view-change consensus time:
\begin{equation}
T_{recovery} = T_{timeout} + T_{vc}
\end{equation}
In our LEO settings, the timeout $T_{timeout}$ is typically configured to accommodate normal network jitter (e.g., $2 \times$ max RTT) to prevent false positives. The view-change consensus time $T_{vc}$ involves a standard BFT process overhead. For intra-plane recovery, considering a plane size of $n=22$ and a hop-by-hop delay, the total recovery time is strictly bounded. Analytically, if a leader fails, the system throughput temporarily drops to zero and autonomously restores within this bounded time window (approximately $O(n \times \text{link\_delay})$). This ensures that even under Byzantine leader failures, the service interruption is transient, satisfying the robustness requirement without indefinite stalling.

\section{Evaluation}
\label{sec:evaluation}

In this section, we conduct simulation experiments to evaluate the performance of our methods. Our evaluation aims to answer the following key questions: 
(1) How much does our communication protocol improve consensus performance on a single plane (Section~\ref{sec:evaluation:intra})?
(2) How can our \protocolname{} improve consensus protocols compared to other native consensus protocols (Section~\ref{sec:evaluation:inter})?

\subsection{Experiment Setup}
We use Hypatia\footnote{https://github.com/snkas/hypatia} simulator to conduct all experiments. 
Hypatia is a high-fidelity simulation platform for large-scale satellite networks, built upon NS-3\footnote{https://www.nsnam.org/}. It accurately models satellite orbital dynamics, dynamic network topology changes, and the communication characteristics of ISLs. 
For the LEO constellation, we utilize the publicly available parameters of the Starlink Phase I (Starlink-I) constellation. This configuration comprises 72 orbital planes, with 22 satellites on each plane. Unless otherwise specified, the default ISL bandwidth is set to 1~Mbps to facilitate the observation of network congestion.

To thoroughly evaluate our protocol, we've divided the comparison into two main categories: native (-native) and optimized (-opt). The native group comprises the standard, unmodified versions of PBFT (PBFT-native) and HotStuff (HotStuff-native). Conversely, the opt group features our proposed protocols, PBFT-opt and HotStuff-opt, which incorporate our novel hierarchical architecture and communication enhancements into their respective native counterparts.

\subsection{Intra-Plane Consensus}
\label{sec:evaluation:intra}

\subsubsection{Bandwidth Storm around Plane Leader}
To demonstrate the ``bandwidth storm'' problem visually, we set a target throughput equal to the peak throughput of the native protocols and force all consensus protocols to operate at this level. We then monitor the utilization of a hot ISL (ISL 0$\rightarrow$1, closer to the leader) and a cold ISL (ISL 10$\rightarrow$11, far away from the leader) over a 50-second interval.

As presented in Fig.~\ref{fig:exp6_1}, our optimized approach can accomplish the same consensus protocol with lower and more balanced bandwidth consumption across different links within the orbital plane. In subfigures (a) and (c), both PBFT-native and HotStuff-native (indicated by red areas) utilize the entire bandwidth of the hot ISL to achieve the target throughput. Conversely, as shown in subplots (b) and (d), their utilization of the cold ISL is less than 10\%. This shows the severe congestion and bandwidth imbalance issues of the native protocols. 
In contrast, our approaches (shown by the blue areas across all four subplots) consume minimal bandwidth, and this consumption is effectively distributed between the hot and cold ISLs. This demonstrates that our optimization not only decreases bandwidth utilization but also achieves load balancing. Furthermore, the unused bandwidth in our optimized protocols suggests a higher throughput potential, which we'll verify in the next experiment.

\begin{figure}[htbp]
\centering
\includegraphics[width=1.0\columnwidth]{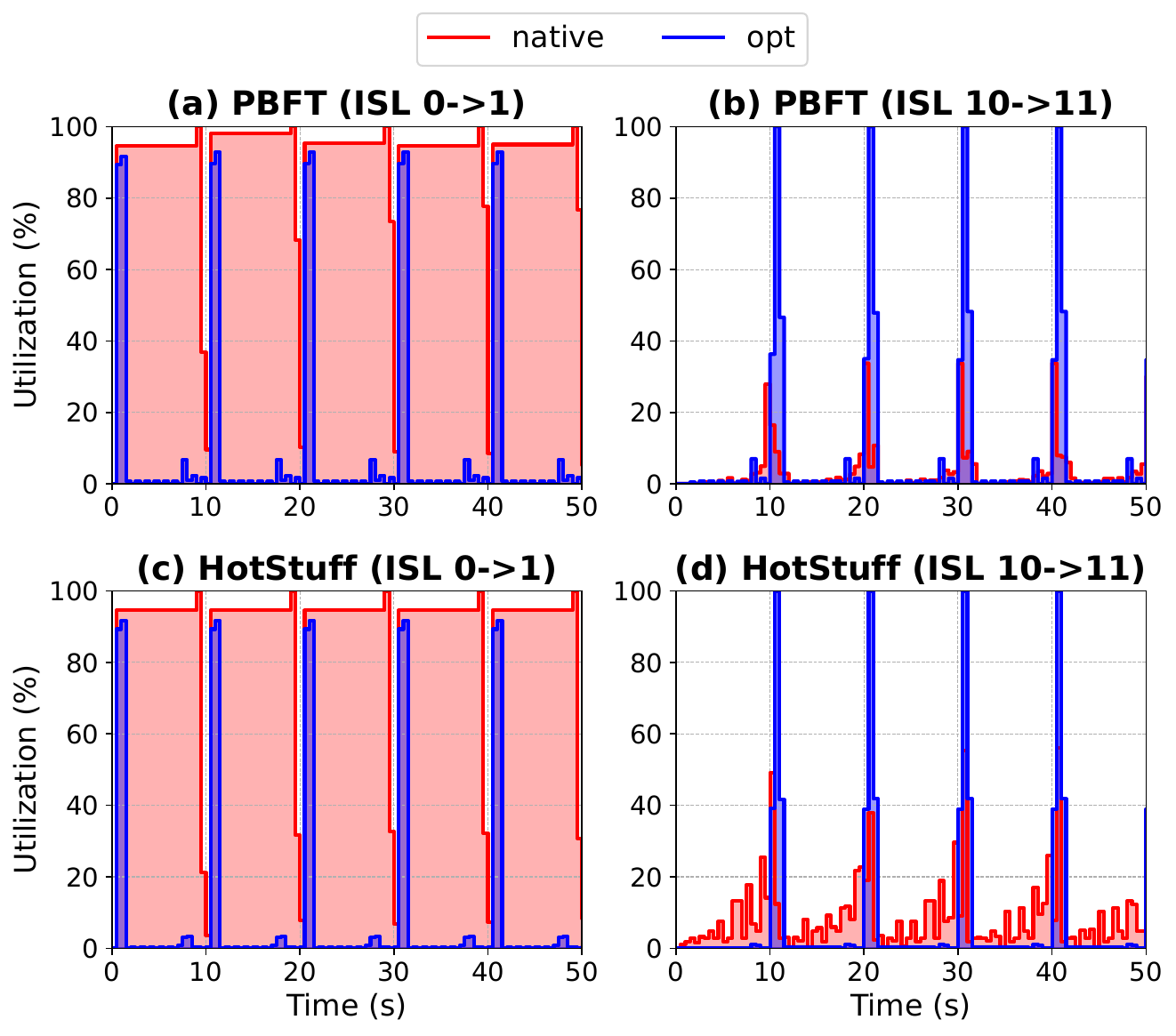}
\caption{ISL utilization at the target throughput. ``x $\rightarrow$ y'' means the ISL from x to y.}
\label{fig:exp6_1}
\end{figure}

\subsubsection{Latency and Throughput under Different Bandwidth}
The previous experiment has shown our optimization's superior bandwidth efficiency for a given workload. This experiment extends that finding by investigating how this efficiency translates into higher throughput when more bandwidth is available. We evaluate the latency and throughput of all four protocols under two bandwidth conditions: 1~Mbps and 10~Mbps. 

As shown in Fig.~\ref{fig:exp6_2a}, at 1~Mbps bandwidth, our approaches achieve a peak throughput of 87 TPS, approximately 10 times higher than the native protocols' 8.2 TPS. 
When bandwidth is increased to 10 Mbps, all protocols show improved performance. Notably, our optimized protocols reach a peak throughput exceeding 870 TPS, still far outperforming the native protocols. This proves that our communication optimization is effective across different bandwidth conditions. 
However, we also observe that as the request rate continues to increase, the system eventually saturates due to computational or protocol-level bottlenecks. This saturation leads to a sharp rise in latency and a plateau or drop in throughput. This observation underscores the necessity for a congestion control mechanism, which will be the subject of our next evaluation.

\begin{figure}[htbp]
\centering
\includegraphics[width=0.9\columnwidth]{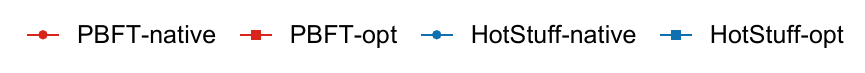}
\centering
\begin{subfigure}[b]{0.48\columnwidth}
    \includegraphics[width=\textwidth]{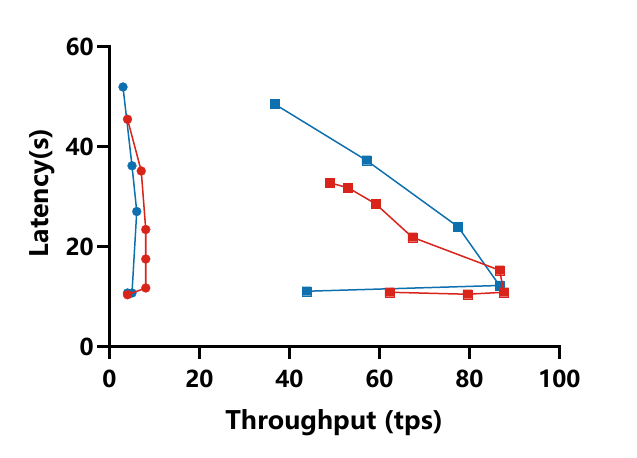}
    \caption{1~Mbps bandwidth}
    \label{fig:exp6_2a}
\end{subfigure}
\hfill
\begin{subfigure}[b]{0.48\columnwidth}
    \includegraphics[width=\textwidth]{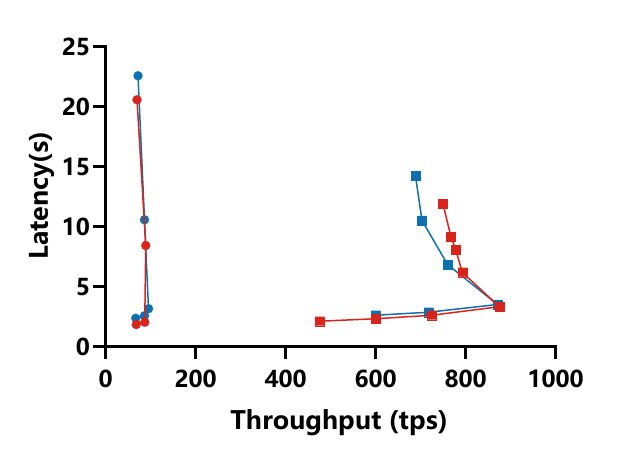}
    \caption{10~Mbps bandwidth}
    \label{fig:exp6_2b}
\end{subfigure}
\caption{The latency and throughput of four consensus protocols under different bandwidth settings.}
\label{fig:exp6_2}
\end{figure}

\subsubsection{Congestion Control}
This experiment validates the effectiveness of a dynamic window mechanism that limits
the total number of consensus instances in the protocol pipeline.
We focus on our optimized protocols, comparing their performance with and without the window mechanism enabled.

As illustrated in Fig.~\ref{fig:exp6_3}, without the window mechanism, the throughput of both our optimized protocols collapses after reaching its peak, a classic sign of congestion collapse. In contrast, when the window mechanism is enabled, the throughput stabilizes at the peak level, and latency is kept under control. This demonstrates that our window mechanism effectively prevents performance degradation from system overload, thus ensuring service stability.

\begin{figure}[htbp]
\centering
\includegraphics[width=0.9\columnwidth]{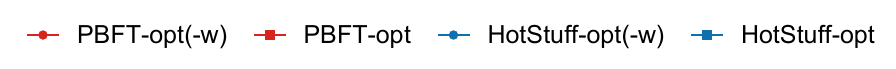}
\centering
\begin{subfigure}[b]{0.48\columnwidth}
    \includegraphics[width=\textwidth]{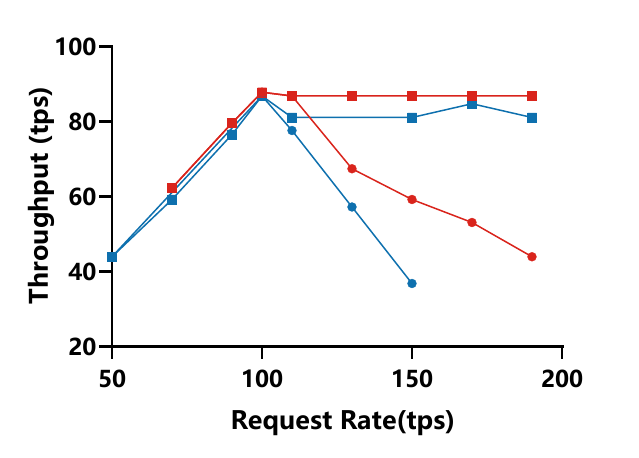}
    \caption{Throughput vs. request rate}
    \label{fig:exp6_3a}
\end{subfigure}
\hfill
\begin{subfigure}[b]{0.48\columnwidth}
    \includegraphics[width=\textwidth]{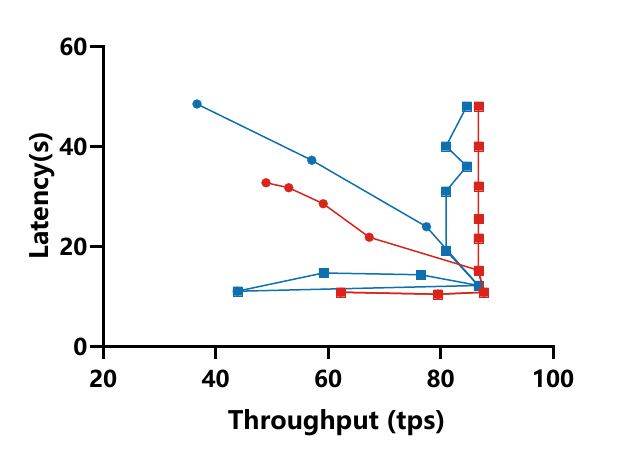}
    \caption{Latency vs. throughput}
    \label{fig:exp6_3b}
\end{subfigure}
\caption{Performance of \protocolname{}-based protocols with and without the dynamic window mechanism. PBFT-opt(-w) and HotStuff-opt(-w) denote PBFT-opt and HotStuff-opt without a window mechanism}
\label{fig:exp6_3}
\end{figure}

\subsubsection{Scalability}

To evaluate the scalability of our communication optimization, we set the number of satellites in the intra-plane consensus group from 7 to 22 and measure the latency and throughput of each consensus protocol.
As shown in Fig.~\ref{fig:exp6_4a}, the native protocols exhibit poor scalability. As the plane size increases from 4 to 22, their maximum throughput drops sharply by 84\%. In contrast, our optimized protocols demonstrate effective scalability, with throughput decreasing by a mere 15\% under the same conditions. This difference is attributed to the message complexity: the $O(N^2)$ message overhead of the native protocols leads to rapidly escalating network contention as $N$ grows, whereas the $O(N)$ complexity of our scheme effectively mitigates this effect. Furthermore, Fig.~\ref{fig:exp6_4b} shows the latency of four consensus protocols. Although the hop-by-hop mechanism introduces more transmission hops, it eliminates the severe network congestion present in the native protocols. Consequently, our optimized protocols not only provide higher throughput but also achieve lower latency, reducing it by up to 27\% compared to the native protocols. These results confirm that our optimizations provide an effective and scalable solution for intra-plane consensus, delivering higher throughput at a lower latency compared to native protocols.

\begin{figure}[htbp]
\centering
\includegraphics[width=0.9\columnwidth]{figures/title1.pdf}
\centering
\begin{subfigure}[b]{0.48\columnwidth}
    \includegraphics[width=\textwidth]{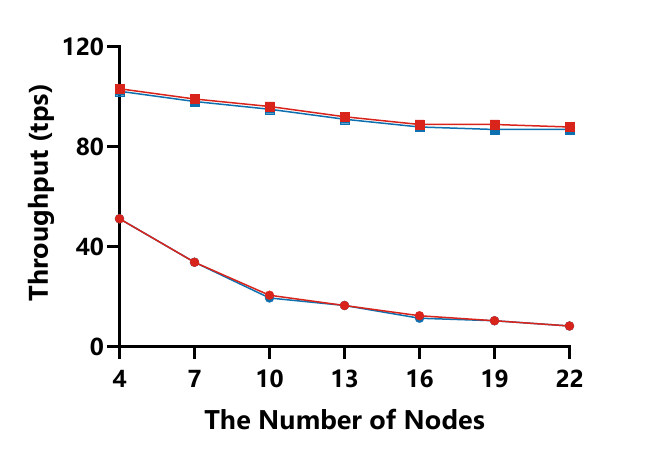}
    \caption{Throughput}
    \label{fig:exp6_4a}
\end{subfigure}
\hfill
\begin{subfigure}[b]{0.48\columnwidth}
    \includegraphics[width=\textwidth]{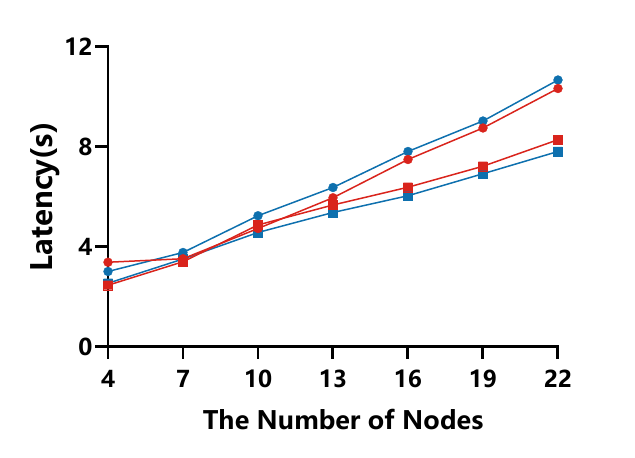}
    \caption{Latency}
    \label{fig:exp6_4b}
\end{subfigure}
\caption{The throughput and latency of intra-plane consensus vs. plane size}
\label{fig:exp6_4}
\end{figure}

\subsection{\protocolname{}}
\label{sec:evaluation:inter}
Having validated the effectiveness of the intra-plane communication optimizations, we evaluate the performance of the complete two-stage architecture across the entire constellation. The primary goal is to verify the advantages of our two-stage protocol over native protocols in terms of scalability and overall performance in a large-scale network environment. 

\subsubsection{LEO Constellation Size}
We progressively increase the total number of nodes participating in consensus and measure the average latency.
As illustrated in Fig.~\ref{fig:exp6_5}, the latency of all protocols increases with the number of nodes, but at different rates. The latency of the native protocols grows much faster than that of our optimized protocols, with PBFT-native exhibiting a particularly rapid increase, faster than HotStuff-native. \protocolname{} architecture successfully decomposes the complexity of global consensus, mitigating the latency pressure from network scaling and showcasing effective scalability.

\subsubsection{Starlink-I Constellation}
We evaluate the latency and throughput of each consensus protocol in the full 1584-node constellation, i.e., the full Starlink-I constellation.
As presented in Fig.~\ref{fig:exp6_6}, the \protocolname{}-based consensus protocol has higher throughput at the same latency, compared to native consensus protocols. On the one hand, PBFT-native did not complete its run within 200 seconds, our set maximum consensus latency, which would require 30 hours of simulation time. On the other hand, HotStuff-native's peak throughput is less than 5 TPS. In contrast, as throughput increases, the latency of the \protocolname{}-based consensus protocol grows very slowly, which confirms that \protocolname{} is an effective approach for achieving high-performance, scalable BFT consensus in a large-scale LEO constellation.

\begin{figure}[htbp]
\centering
\includegraphics[width=0.9\columnwidth]{figures/title1.pdf}
\centering
\begin{subfigure}[b]{0.48\columnwidth}
    \centering
    \includegraphics[width=\textwidth]{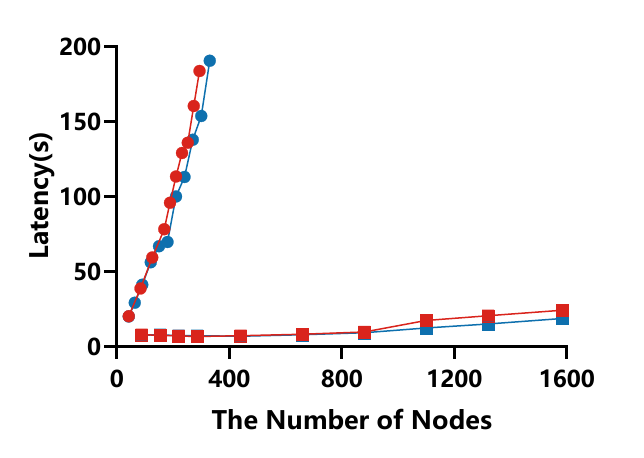}
    \caption{Latency vs. network size}
    \label{fig:exp6_5}
\end{subfigure}
\hfill
\begin{subfigure}[b]{0.48\columnwidth}
    \centering
    \includegraphics[width=\textwidth]{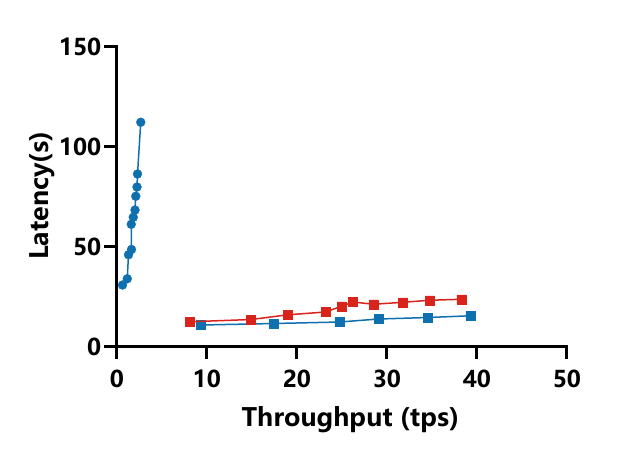}
    \caption{Latency vs. throughput}
    \label{fig:exp6_6}
\end{subfigure}
\caption{Consensus latency and throughput in the full-scale constellation}
\label{fig:combined_6_5_6_6}
\end{figure}

\section{Conclusion and Future Work}
\label{sec:conclusion}
This paper presents OrbitBFT, a hierarchical BFT protocol for large-scale LEO satellite constellations. It decomposes consensus into intra- and inter-plane phases, exploiting stable intra-plane topology to reduce communication overhead and improve scalability. OrbitBFT includes a Byzantine-resilient bypass mechanism and a congestion-aware hop-by-hop scheme to handle unreliable links and bandwidth limits. It also integrates optimized PBFT and HotStuff, achieving linear message complexity with improved throughput and latency. Simulations based on real-world Starlink parameters demonstrate its effectiveness and scalability. Future work focuses on enhancing robustness against localized Byzantine partitions, where colluding adversaries form connected subgraphs that challenge quorum formation and routing adaptability.



\bibliographystyle{IEEEtran}
\bibliography{ref} 

\end{document}